\documentclass{article}
\usepackage{graphicx} % Required for inserting images
\usepackage{helvet}
\usepackage{amsmath}
\usepackage{amssymb}
\usepackage{amsthm}
\usepackage{mathtools}
\usepackage{hyperref}
\usepackage[svgnames]{xcolor}
\usepackage{graphicx}
\graphicspath{{./images/}}
\usepackage{geometry}
\usepackage{datetime}
\usepackage[shortlabels]{enumitem}

\newcommand{\F}{\mathbb{F}}
\newcommand{\Z}{\mathbb{Z}}
\newcommand{\fp}{(e_1,e_2)}
\newcommand{\rop}{(\alpha + u_1\tau,\alpha + u_2\tau)}
\newcommand{\specialcell}[2][c]{%
\begin{tabular}[#1]{@{}c@{}}#2\end{tabular}}
\newtheorem{theorem}{Theorem}[section]
\newtheorem{corollary}[theorem]{Corollary}
\newtheorem{lemma}[theorem]{Lemma}

\newtheorem{proposition}[theorem]{Proposition}

\theoremstyle{definition}

\theoremstyle{definition}
\newtheorem{definition}[theorem]{Definition}

\theoremstyle{remark}
\newtheorem{remark}[theorem]{Remark}

\title{ New Quantum MDS Codes with Flexible Parameters from Hermitian Self-Orthogonal GRS Codes }
\author{Oisin Campion\footnote{School of Mathematics and Statistics, University College Dublin, Ireland}, Fernando Hernando\footnote{Instituto Universitario de Matem\'aticas y Aplicaciones de Castell\'on and Departamento de Matem\'aticas, Universitat Jaume I, Campus de Riu Sec, 12071 Castell\'o, Spain},   Gary McGuire\footnote{School of Mathematics and Statistics, University College Dublin, Ireland}}
\date{\today}

\begin{document}

\maketitle

\begin{abstract}
    Let $q$ be a prime power. 
  Let  $\lambda>1$ be a divisor of $q-1$, and 
let $\tau>1$ and  $\rho>1$ be divisors of $q+1$. 
    Under certain conditions we prove that there exists an MDS stabilizer quantum code with length 
    $n=\lambda \tau \sigma$ where $2\le \sigma \le \rho$.
    This is a flexible construction, which
    includes new MDS parameters not known before.
\end{abstract}

\section{Introduction}

In this paper we will prove the following theorem.

\begin{theorem}
Let $q\ge 3$ be a prime power. Let $\lambda>1$ be a divisor of $q-1$, and 
let $\tau>1$ and  $\rho>1$ be divisors of $q+1$. 
Assume that $\gcd (\lambda,\tau)=1$. 
Let  $\kappa=\gcd (\lambda,\rho)\cdot  \gcd(\tau,\rho)$ and assume that $ \frac{\rho}{\kappa} \ge 2$.
Let $\sigma$ be any integer with $ \frac{\rho}{\kappa} \ge  \sigma \ge 2$.
Let $n = \lambda\tau\sigma$.
Let $T$ be chosen according to this table:
\begin{center}
\begin{tabular}{||c|c||} 
 \hline
  Conditions &  $T$ \\  
 \hline\hline
  $\lambda$ even &  $\frac{\lambda+4\tau}{2}$\\ 
 \hline
  \specialcell{$\lambda$ odd, and either $\lambda<\tau$, \\ $\tau$ even or $\rho=2$} &  $\lambda + \tau$   \\
 \hline
  $ \lambda $ odd, $\lambda >\tau, \tau $ odd, $\rho \neq 2$ &  $\frac{\lambda+3\tau}{2}$ \\

 \hline
\end{tabular}
\end{center}

Then for any $d$ with  $2\le d \le T$ there exists a $[[n,n-2d+2,d]]_q$ quantum MDS  code. 
\end{theorem}

Stabilizer quantum error-correcting codes have been studied by many authors because they can be constructed  from classical additive codes  in $\mathbb{F}_q^{2n}$ which are self-orthogonal with respect to a  Hermitian form. 
In particular, stabilizer codes can be obtained from suitable Hermitian self-orthogonal classical linear codes (see \cite{Ketkar} or \cite{Calderbank,AK} for details). 
In this article we will utilize this construction -- we use self-orthogonal generalized Reed-Solomon codes to 
 construct stabilizer quantum codes.

Quantum MDS codes are those achieving the quantum Singleton bound, 
so an MDS $[[n,k,d]]_q$ code has $n+2=k+2d$.
There are many papers on these types of codes (some recent papers are \cite{Fang,Ball,Liu-LiuX}). 
The MDS conjecture limits the length of a $q$-ary quantum MDS code to be at most $q^2+2$ (\cite{Ketkar}). MDS codes with length
smaller than $q+1$ are already known, so researchers have
recently concentrated on lengths between $q+1$ and $q^2+2$.
Our main theorem constructs new quantum MDS codes 
with lengths in this region. Many papers construct codes whose length
is either a multiple of $q-1$ or $q+1$; we  construct codes whose length 
 is not a multiple of $q-1$ or $q+1$.
It has been shown in \cite{PhysRevA.77.012308} that $d\le q+1$ for codes that are constructed with this GRS method,
and separately, 
the article \cite{JLLX2010} constructs many codes with $d<q/2$, so recently papers have 
concentrated on constructing codes with $ q/2 \le d \le q+1$.
We present some new codes whose length 
 is not a multiple of $q-1$ or $q+1$ and which have $d > q/2$.
 
In a recent paper \cite{Barbero24} we investigated codes with
length larger than $q^2+2$ by a similar method.

The paper is laid out as follows.
After the preliminaries in Section \ref{prelim}, we present our construction in 
Section \ref{section_code_construction}.
We use the Hermitian  construction with GRS codes.
Previous works using a twist vector have proved the existence of a twist
vector with the required properties, whereas a feature of our construction
is that we can be very explicit in the choice of twist vector.
In Sections \ref{failurepoints} and \ref{section_first_failure_point}
we prove our main theorem by a careful analysis of when self-orthogonality 
of our codes can be guaranteed.
Our main theorem is presented in Section \ref{mainsection} as Theorem \ref{main}.
In Section \ref{examples} we present some applications of Theorem \ref{main},
constructing MDS codes with parameters that are new.

\section{Preliminaries}\label{prelim}
Let $q$ be a prime power, with $q \ge 3$. We denote by $\F_q$ the field with $q$ elements. For two vectors $\boldsymbol{a}=(a_0,\ldots a_{n-1}), \boldsymbol{b}=(b_0,\ldots,b_{n-1})$ in $\F_{q^2}^n$, we define their Hermitian inner product to be
\begin{equation*}
    \boldsymbol{a}\cdot_h\boldsymbol{b} = \sum_{i=0}^{n-1}a_ib_i^q,
\end{equation*}
and their Euclidean inner product as
\begin{equation*}
    \boldsymbol{a}\cdot_e\boldsymbol{b} = \sum_{i=0}^{n-1}a_ib_i.
\end{equation*}
We use the symbol $\perp_h$ (respectively, $\perp_e$) to denote orthogonality under the Hermitian (respectively, Euclidean) product. For a vector subspace (code) $C$ of $\F_{q^2}^{n}$, we write $C^{\perp_h}$ (respectively, $C^{\perp_e}$) to denote the dual code with respect to the Hermitian (respectively, Euclidean) inner product. We will denote the minimum distance of $C$ by $d(C)$. 

For a non-negative integer $s$, and a codeword $\boldsymbol{c}=(c_0,\ldots,c_{n-1}) \in C$, we define $\textbf{c}^s := (c_0^s,\ldots,c_{n-1}^s)$, and 
\begin{equation*}
    C^s := \{ \boldsymbol{c}^s : \boldsymbol{c}\in C\} \subseteq \F_{q^2}^n.
\end{equation*}
We say that two codes are isometric if there exists a bijection between them that preserves Hamming weights. 

\begin{proposition}{(Classical Singleton Bound)}
    If a classical linear $[n,k,d]$ code exists, then 
    the parameters satisfy $k+d\le n+1$.
\end{proposition}

For quantum codes there is also a Singleton bound.

\begin{proposition}{(Quantum Singleton Bound)}
    If a $[[n,k,d]]_q$ stabilizer quantum code exists, then the parameters satisfy 
    $k+2d\le n+2$.
    
\end{proposition}
Codes that attain this bound are called quantum MDS codes. 
\begin{theorem}\label{thm_quantcode}
        Let $C$ be a linear $[n,k,d]$ error-correcting code over the field $\F_{q^2}$ such that $C\subseteq C^{\perp_h}$. Then, there exists an $[[n,n-2k,\geq d^{\perp_h}]]_q$ stabilizer quantum code, where $d^{\perp_h}$ stands for the minimum distance of $C^{\perp_h}$.
\end{theorem}
The goal of this paper is to construct codes satisfying Theorem \ref{thm_quantcode}. Our general framework is that of evaluation codes. Let $\boldsymbol{A}\in \F_{q^2}^n$ be an ordered list of distinct points, $\boldsymbol{A} = (a_0,\ldots a_{n-1})$, called the \textit{evaluation set}. We choose another vector $\boldsymbol{v}=(v_0\ldots,v_{n-1}) \in (\F_{q^2}^*)^n$, referred to as the \textit{twist vector}. For $1 \le k \le n$, we denote $\F_{q^2}[X]_{<k}$ as the $k$-dimensional vector subspace of $\F_{q^2}[X]$, consisting of all polynomials over $\F_{q^2}$ of degree less than $k$. Our code is obtained as the image of the linear evaluation map: 
\begin{equation*}
    ev_{\boldsymbol{v},\boldsymbol{A}} :\F_{q^2}[X]_{<k} \longrightarrow \F_{q^2}^n, \ f \mapsto (v_0f(a_0),\ldots, v_{n-1}f(a_{n-1}))
\end{equation*}
This map is injective, and provides the following class of codes. 
\begin{definition}\label{def_GRSC}
    The \textit{generalized Reed-Solomon code} (GRSC) $\text{GRS}_{n,k}(\boldsymbol{v}, \boldsymbol{A})$ is the image of the above evaluation map, that is 
    \begin{equation*}
        \text{GRS}_{n,k}(\boldsymbol{v},\boldsymbol{A}):= ev_{\boldsymbol{v},\boldsymbol{A}}(\F_{q^2}[X]_{<k}) = \text{span}\{ev_{\boldsymbol{v},\boldsymbol{A}}(X^e):0\le e\le k-1\} \subseteq \F_{q^2}^n.
    \end{equation*}
\end{definition}
\begin{remark}\label{remark_GRSC}
   The length and the dimension of a GRSC, $\text{GRS}_{n,k}(\boldsymbol{v},\boldsymbol{A})$ are $n$ and $k$, respectively. Since these codes are MDS, their minimum distance is equal to $n-k+1$. Is is a standard fact that the Euclidean dual code $(\text{GRS}_{n,k}(\boldsymbol{v},\boldsymbol{A}))^{\perp_e}$ is another MDS GRSC. Moreover, 
\begin{equation*}
    (\text{GRS}_{n,k}(\boldsymbol{v},\boldsymbol{A}))^{\perp_h} = ((\text{GRS}_{n,k}(\boldsymbol{v},\boldsymbol{A}))^{\perp_h}) ^q
\end{equation*}
which implies that $(\text{GRS}_{n,k}(\boldsymbol{v},\boldsymbol{A}))^{\perp_h} $ and $(\text{GRS}_{n,k}(\boldsymbol{v},\boldsymbol{A}))^{\perp_e}$ are isometric. Therefore,\\ $(\text{GRS}_{n,k}(\boldsymbol{v},\boldsymbol{A}))^{\perp_h} $ is also MDS.  
\end{remark}

\section{Code Construction}\label{section_code_construction}
Let $q$ be a prime power, $q\ge3$. Let $\lambda>1$ be a divisor of $q-1$, and 
let $\tau>1$ and  $\rho>1$ be divisors of $q+1$. 
We  assume that $\gcd (\lambda,\tau)=1$.
We let  $\kappa_1=\gcd (\lambda,\rho), \kappa_2= \gcd(\tau,\rho)$, and $ \kappa=\kappa_1\kappa_2$.
%Later we will have two separate cases, when $\rho=2$ and when $\rho>2$.
%When $\rho=2$ then $d=1$ if $\lambda$ is odd, and $d=2$ if $\lambda$ is even. 
We assume for the rest of this paper that $ \frac{\rho}{\kappa} \ge 2$, and 
we let $\sigma$ be any integer with $ \frac{\rho}{\kappa} \ge  \sigma \ge 2$.

Let $\zeta_t$ denote a primitive $t$-th root of unity.

The evaluation set for our code is: 
\begin{equation}
    \boldsymbol{A} := \{\zeta_\lambda^{i}\zeta_{\tau}^{j}\zeta_{\rho}^{k} : 
    0\leq i < \lambda, \ 0\leq j < \tau,\
    0\leq k < \sigma\}
\end{equation}

\begin{lemma}
    Assume the notation above. Then the elements of $\boldsymbol{A}$ are distinct.
\end{lemma}

\begin{proof}
    It suffices to show that if $\zeta_\lambda^{i}\zeta_{\tau}^{j}\zeta_{\rho}^{k} =1$
    then $i=j=k=0$.
    Suppose $\zeta_\lambda^{i}\zeta_{\tau}^{j}\zeta_{\rho}^{k} =1$.
    Then $\zeta_\lambda^{i\tau} \zeta_{\rho}^{k\tau} =1$. Since 
    $\gcd (\lambda,\tau)=1$, $\zeta_\lambda^\tau$ is a primitive $\lambda$-th root of unity, call it $\alpha$.
    Since 
    $\gcd (\rho,\tau)=\kappa_2$, $\zeta_\rho^\tau$ is a primitive $\rho/\kappa_2$-th root of unity, call it $\beta$. So $\alpha^i \beta^k =1$. Taking this to the power of $\lambda$ we 
    get $\beta^{k\lambda}=1$ which implies $\rho | k \lambda \kappa_2$.
    Recall that $k<\sigma \le \frac{\rho}{\kappa}$, which implies that 
    $k\lambda \kappa_2 < \frac{\rho\lambda}{\kappa_1}= lcm (\rho, \lambda)$.
    So $\rho$ divides $k\lambda \kappa_2$, and $k\lambda \kappa_2$ is a multiple of $\lambda$ that is 
    smaller than $lcm (\rho, \lambda)$. Therefore $k=0$.
    
    A similar argument shows that $i=0$ and $j=0$.
\end{proof}

Each point in the evaluation set is associated with a unique integer triple $(i,j,k)$, 
where $0\leq i < \lambda, \ 0\leq j < \tau,\
    0\leq k < \sigma$,
and we will order the points in our evaluation set according to the lexicographic ordering on these triples, writing them as $\boldsymbol{A}(i,j,k):=\zeta_\lambda^{i}\zeta_{\tau}^{j}\zeta_{\rho}^{k}$.

Next, we select $s_0,\ldots s_{\sigma-1}\in \F_q^*$ in the following way: 
\begin{itemize}
    \item If $\sigma=2$, then set $s_0=1, s_1=-1$. 
    \item Otherwise, set $s_0,\ldots,s_{\sigma-3} =1$, select $s_{\sigma-2} \in \F_q \setminus \{0,-(s_0+\ldots +s_{\sigma-3})\}$ and set $s_{\sigma-1} = -(s_0+\ldots s_{\sigma-2})$. 
\end{itemize}

This ensures that $\sum_{k=0}^{\sigma-1}s_k=0$ (which in fact is the only important point, any nonzero values for $s_k$ can be chosen that satisfy this).

We are now going to define our twist vector. 
We order the elements of the twist vector $\boldsymbol{v}$ in the same way as the evaluation set, with each coordinate labelled by a unique integer triple $(i,j,k)$.
In fact, the ordering does not matter as long as the twist vector $\boldsymbol{v}$ uses the
same ordering as the evaluation set $\boldsymbol{A}$.

Let $L$ be an integer. This $L$ will be a parameter that we can fix later.

We select the twist vector $\boldsymbol{v}$ so that the $(q+1)$-power of $\boldsymbol{v}$ has 
$\zeta_{\lambda}^{-i L}\cdot  s_{k}$ in the coordinate labelled $(i,j,k)$.
We will write $\boldsymbol{v}(i,j,k)^{q+1}:= \zeta_{\lambda}^{-i L}  s_{k}$.
Note that $\zeta_{\lambda}^{-i L}  s_{k}$ does not depend on $j$, so 
$\boldsymbol{v}^{q+1}$ has
the same entry for each value of $j$.
Since the elements of $\boldsymbol{v}^{q+1}$ lie in $\F_q$,
it is always possible to find a suitable $\boldsymbol{v}$ with coefficients in $\F_{q^2}$, as we now prove. 

\begin{lemma}
Given any $c\in \F_q$ the polynomial $x^{q+1}-c$ has all its roots in $\F_{q^2}$.
\end{lemma}

\begin{proof}
Since the $q+1$ roots of unity lie in $\F_{q^2}$ it suffices to show that
$x^{q+1}-c$ has one root in $\F_{q^2}$.
This follows from the fact that the norm map from $\F_{q^2}$
to $\F_q$ is surjective.
\end{proof}

\bigskip

We now seek to understand the conditions under which  our evaluation code will be contained in its Hermitian dual. 

\begin{theorem}\label{thm_OC}
    Let $X^{e_1},X^{e_2}$ be two monomials. Then the evaluation vectors for these two monomials are orthogonal under the Hermitian inner product if any one of the following conditions holds: 
    \begin{enumerate}
    \item $e_1+e_2  \not\equiv L $ (mod $\lambda$).
    \item $e_1 \not\equiv e_2$ (mod $\tau$).
    \item $e_1\equiv e_2$ (mod $\rho$).
\end{enumerate}
\end{theorem}
\begin{proof}
    The Hermitian inner product of the evaluation vectors is given by

\begin{equation*}
\begin{split}
    ev_{\boldsymbol{v},\boldsymbol{A}}\left(X^{e_1}\right)\cdot_h ev_{\boldsymbol{v},\boldsymbol{A}}\left(X^{e_2}\right) = & \sum_{i,j,k}\boldsymbol{v}(i,j,k)^{q+1} \boldsymbol{A}(i,j,k)^{e_1 +qe_2}\\
    =&\sum_{i,j,k}\zeta_{\lambda}^{-i L}  s_{k} \ (\zeta_\lambda^{i}\zeta_{\tau}^{j}\zeta_{\rho}^{k})^{e_1+qe_2}\\
    =&\left(\sum_{i=0}^{\lambda-1}\zeta_{\lambda}^{i(e_1+qe_2-L)}\right)\left(\sum_{j=0}^{\tau-1}\zeta_{\tau}^{j(e_1+qe_2)}\right)\left(\sum_{k=0}^{\sigma-1}s_{k}\zeta_{\rho}^{k(e_1+qe_2)}\right).\\
\end{split}  
\end{equation*}

We examine each of the terms: 
\begin{itemize}
    \item The first term is zero $\iff$ $e_1+qe_2-L\not\equiv 0$ (mod $\lambda$) $\iff$ 
    $e_1+e_2  \not\equiv L $ (mod $\lambda$). 
    \item The second term is zero $\iff$ $e_1+qe_2\not\equiv 0$ (mod $\tau$) 
    $\iff$ $e_1 \not\equiv e_2$ (mod $\tau$).
    \item The third term is zero if $e_1+qe_2\equiv0$ (mod $\rho$), because in that case
    the sum is $\sum_{k=0}^{\sigma-1}s_k$, which is 0. Therefore the third term is 0 if $e_1 \equiv e_2$ (mod $\rho$).
\end{itemize}
\end{proof}

In the next sections we investigate when two monomials fail to be orthogonal.

\section{Failure Points}\label{failurepoints}

We continue the notation of the previous section.

If the evaluation vectors for the monomials $X^{e_1},X^{e_2}$ are not orthogonal, then all three conditions in Theorem \ref{thm_OC} fail to hold. We are interested in knowing exactly when that happens. 

\begin{definition}
    Let $X^{e_1},X^{e_2}$ be two monomials. If all three of the following conditions hold, then we call  the point $\left(e_1,e_2\right)$ a \textbf{failure point}: 
        \begin{enumerate}
    \item $e_1+e_2  \equiv L $ (mod $\lambda$).
    \item $e_1 \equiv e_2$ (mod $\tau$).
    \item $e_1\not\equiv e_2$ (mod $\rho$).
\end{enumerate}
\end{definition}
We now gather some results about failure points. 

\begin{lemma}
    If $\fp$ is a failure point, then $\left(e_2,e_1\right)$ is a failure point. 
\end{lemma}
\begin{proof}
    This follows from the symmetry of the conditions
\end{proof}

\begin{lemma}
    If $\fp$ is a failure point then $e_1\neq e_2$. 
\end{lemma}
\begin{proof}
    Follows from condition 3. 
\end{proof}

\begin{remark}
    Using these
    lemmas, we will assume that a failure point $\fp$ satisfies $e_2 > e_1$. 
\end{remark}

\begin{lemma}\label{lemmarop}
    A failure point $\fp$ can be written uniquely as $\rop$ where $0\leq \alpha \leq \tau-1$. 
\end{lemma}
\begin{proof}
    Using the division algorithm, write $e_1 = \alpha + u_1\tau$, $e_2 = \alpha' + u_2\tau$. By condition 2, we can deduce that $\alpha = \alpha'$. 
\end{proof}
\begin{definition}
    The \textbf{first failure point} is the failure point $\left(e_1,e_2\right)$ such that for any other failure point $\left(e_1',e_2'\right)$ we have $e_2 <e_2'$.
\end{definition}
\begin{lemma}\label{lemma_gap1}
    Let $\fp$ be the first failure point. Then $e_2-e_1 \in \{\tau,2\tau\}$. Moreover, if $\rho=2$ then $e_2-e_1=\tau$. 
\end{lemma}
\begin{proof}
    Suppose otherwise for a contradiction. Using \ref{lemmarop}, write the first failure point as $\rop$. There are two cases:
    \begin{itemize}
        \item If $u_1+u_2$ is odd, consider the point 
        \begin{equation*}
        \begin{split}
\left(e_1',e_2'\right) & =  \left(\alpha +u_1\tau + \frac{u_2-u_1-1}{2}\tau,\alpha +u_2\tau -\frac{u_2-u_1-1}{2}\tau\right)\\
& = \left(\alpha + \frac{u_1+u_2-1}{2}\tau,\alpha + \frac{u_1+u_2+1}{2}\tau\right).  
        \end{split}
        \end{equation*}

It is easy to check that $\left(e_1',e_2'\right)$ is another failure point, and that $e_2'<e_2$. This contradicts $\left(e_1,e_2\right)$ being the first failure point. 
\item If $u_1+u_2$ is even, consider the point 
        \begin{equation*}
        \begin{split}
        \left(e_1',e_2'\right) &= \left(\alpha +u_1\tau + \frac{u_2-u_1-2}{2}\tau,\alpha +u_2\tau -\frac{u_2-u_1-2}{2}\tau\right) \\
        &= \left(\alpha + \frac{u_1+u_2-2}{2}\tau,\alpha + \frac{u_1+u_2+2}{2}\tau\right).
        \end{split}
        \end{equation*}

Again, notice that $\left(e_1',e_2'\right)$ is another failure point, and that $e_2'<e_2$, which is a contradiction. 

Finally, notice that having $\rho=2$ precludes the case of $u_1+u_2$ being even, and the result follows.  
    \end{itemize}
\end{proof}

By Lemma \ref{lemma_gap1}, we know that the first failure point will lie on one of the lines $y=x+\epsilon \tau$, where $\epsilon =1$ or 2.
Moreover, the cases $\rho=2$ and $\rho>2$ are slightly different.
When $\rho =2$ we know for sure that $\epsilon =1$.

\section{Computing the First Failure Point}\label{section_first_failure_point}

We now analyze a variety of cases. We will denote the coordinates of the first failure point $\left(T_1,T_2\right)$. In each case, we also have freedom to choose $L$, so we will find the value of $L$ that maximises $T_2$. Letting $T:= T_2+1$, we will be able to construct a QEC with minimum distance $T$. 

In each case, using Lemma \ref{lemma_gap1}, we will write the first failure point as $\left(\gamma,\gamma+\epsilon\tau\right)$, where $\gamma \ge 0$ and $\epsilon \in \{1,2\}$. If $\epsilon$ has been determined, then the task of maximising $T_2$ is equivalent to maximising the quantity $e_1+e_2$.

\subsection[Case 1: lambda even.]{Case 1: $\lambda$ even.}
Since $(\lambda,\tau)=1$, it follows that $\tau$ is odd. 
Also, it follows from $ \frac{\rho}{\kappa} \ge 2$ that $\rho>2$.
Our candidates for $L$ are the elements of the additive group $\Z/\lambda\Z$.
We look at two cases: 
\begin{itemize}
    \item Suppose $L$ lies in the subgroup $\langle 2\rangle$ and consider the representatives $\{2\tau,2\tau-2,\ldots 2\tau-(\lambda-2)\}$. We can write this as $L=2\tau-\delta$ for some $\delta \in \{0,2,4,\ldots,\lambda-2\} $. The equation for condition 1 becomes: 
    \begin{equation*}
        2\gamma + \epsilon\tau = t\lambda + 2\tau-\delta.
    \end{equation*}
    Examining parities, it is clear that $\epsilon =2$. The equation reduces to 
    \begin{equation*}
        2\gamma = t\lambda -\delta
    \end{equation*}
    When $\delta=0$, the smallest $t$ for which there is a solution is $t=0$, and the first failure point will be $(0,2\tau)$. Otherwise, the smallest $t$ is $t=1$. It is optimal to then choose $L$ as large as possible; in this case, choosing $L=2\tau-2$. It follows that $\gamma = \frac{\lambda-2}{2}$, and the first failure point will be $\left(\frac{\lambda-2}{2}, \frac{\lambda+4\tau-2}{2}\right)$.

    \item Suppose $L$ lies in the coset $\langle2\rangle+1$ and consider the set of representatives $\{\tau-\delta\}$ for $\delta \in \{0,2,4,\ldots,\lambda-2\}$. The equation becomes 
        \begin{equation*}
        2\gamma + \epsilon\tau = t\lambda + \tau-\delta.
    \end{equation*}
    Examining parities, it is clear that $\epsilon =1$. The equation reduces to 
        \begin{equation*}
        2\gamma = t\lambda -\delta
    \end{equation*}
    There is no solution for $t \leq 0$, unless $\delta=0$, in which case the first failure point will be $(0,\tau)$. Otherwise, there is a solution for $t=1$; we choose $L=\tau-2$ and it follows that the first failure point will be $\left(\frac{\lambda-2}{2},\frac{\lambda+2\tau-2}{2}\right)$
\end{itemize}
Analysing these possibilities, the best choice is $L = 2\tau-2$, the resulting first failure point will be 
\[
\boxed{\left(T_1,T_2\right) = \left(\frac{\lambda-2}{2},\frac{\lambda+4\tau-2}{2}\right)}
\]
and the greatest possible distance is
\[
\boxed{T = \frac{\lambda+4\tau}{2}}
\]

\subsection[Case 2: lambda odd with certain  conditions]{Case 2: $\lambda$  odd with certain  conditions}

Assume that $\lambda$ is odd, and at least one of the following is true: 
\begin{itemize}
    \item $\lambda < \tau$
    \item $\tau$ even
    \item $\rho=2$.
\end{itemize}
For each of the above points, we need a separate proof. 
\subsubsection[Case 2a]{$\lambda$ odd, $\lambda<\tau$}
Consider the set of monomial powers: 
\[
S = \{0,1,\ldots,\tau-1,\tau,\ldots,\tau+\lambda-1\}.
\]
Since $\lambda< \tau$, we have that every element of $S$ is strictly less than $2\tau$. Thus, the only points satisfying condition 2 are: 
\[
\left(0,\tau\right), \left(1,\tau+1\right),\ldots , \left(\lambda-1,\tau+\lambda-1\right). 
\]
To check condition 1, we are interested in the sums $e_1+e_2-L$. For each of the points above, they are: 
\[
\tau-L, \tau-L+ 2, \tau-L +4, \ldots, \tau -L + 2\left(\lambda-1\right).
\]
To avoid these sums being divisible by $\lambda$ for as long as possible, it is optimal to select $L$ such that $L \equiv \tau-2 \text{ (mod } \lambda)$, since $\lambda$ is odd. Now, to see the first failure point, write the points above in the form $\left(\gamma,\gamma+\tau\right)$. The sum becomes $2\left(\gamma+1\right)$, and so the smallest non-negative $\gamma$ for which this is divisible by $\lambda$ is $\gamma = \lambda-1$. 

So, the first failure point in this case is 
\[
\boxed{\left(T_1,T_2\right) = \left(\lambda-1,\lambda+\tau-1\right)}
\]
and the greatest possible distance is 
\[
\boxed{T = \lambda+\tau}
\]

\subsubsection[Case 2b]{$\lambda$ odd, $\tau$ even}
Write the first failure point as $\left(\gamma,\gamma+\epsilon\tau\right)$, where $\gamma \ge 0$ and $\epsilon \in \{1,2\}$. 
It follows from $ \frac{\rho}{\kappa} \ge 2$ that $\rho>2$, although we will not need this in this section
because we will prove $\epsilon=1$. 

Suppose for a contradiction that the first failure point can be written as $(e_1,e_2)=(\gamma,\gamma+2\tau)$. Consider the point 
\[
\left(e_1',e_2'\right) = \left(\gamma +\frac{\tau}{2},\gamma +2\tau-\frac{\tau}{2}\right).
\]
Since $\tau$ is even, this is an integer point. It is easy to check that $\left(e_1',e_2'\right)$ is another failure point, and that $e_2'<e_2$. This contradicts $\left(e_1,e_2\right)$ being the first failure point.

We now proceed to compute the optimal $L$, and the first failure point. We consider as our choices for $L$ the representatives $\tau-\delta$, for $\delta \in \{0,1,\ldots , \lambda-1\}$. Condition 1 becomes
\[ 2\gamma= t\lambda -\delta
\]
The term on the LHS is always even. We consider two possibilities: 
\begin{itemize}
    \item If $\delta$ is odd, then the smallest $t$ for which there is an integer solution is $t=1$. The optimal choice for $L$ would thus be $\tau-1$, and the first failure point will be $\left(\frac{\lambda-1}{2},\frac{\lambda+2\tau-1}{2}\right)$.
    \item If $\delta=0$ then the smallest solution is $t=0$, and the first failure point will be $(0,\tau)$. 
    \item If $\delta \neq 0$ is even, then the smallest solution is $t=2$, the optimal choice for $L$ is $\tau -2$, and the resulting first failure point is $(\lambda-1,\lambda+\tau-1)$
\end{itemize}
The optimal choice for $L$ in this case is therefore $L = \tau-2$. Then the first failure point will be 
\[
\boxed{\left(T_1,T_2\right) = \left(\lambda-1,\lambda+\tau-1\right)}
\]
and the greatest possible minimum distance of the code is 
\[
\boxed{T = \lambda+\tau}
\]

\subsubsection[Case 2c]{$\lambda$ odd, $\rho=2$}
By Lemma \ref{lemma_gap1}, the first failure point can be written as $(\gamma,\gamma + \tau)$. We consider as our choices for $L$ the representatives $\tau-\delta$, for $\delta \in \{0,1,\ldots , \lambda-1\}$. Condition 1 becomes
\[ 2\gamma= t\lambda -\delta
\]
The analysis is now identical to the previous case. 
The optimal choice for $L$ in this case is therefore $L = \tau-2$. Then the first failure point will be 
\[
\boxed{\left(T_1,T_2\right) = \left(\lambda-1,\lambda+\tau-1\right)}
\]
and the greatest possible minimum distance of the code is 
\[
\boxed{T = \lambda+\tau}
\]
\subsection[Case 3: lambda  odd, all remaining cases]{Case 3: $ \lambda $ odd, all remaining cases}

We now consider all remaining cases that are not covered by the previous sections.
Assume that $\lambda$ is odd and that
$\lambda >\tau$, and that $\tau $ is odd, and that $\rho > 2$.
We first prove the following technical lemma, which will simplify the analysis of this case. 

\begin{lemma}\label{worstcase}
    Suppose $\lambda > \tau$. Let $Q,Q'$ be positive integers with $Q'>Q$. Consider the four lines: 
    \begin{equation*}
        \begin{split}
            & L_1: y=-x+Q\lambda +L.\\
            & L_2: y=-x+Q'\lambda +L.\\
            & L_1': y=x + \tau.\\
            & L_2': y=x + 2\tau.\\
        \end{split}
    \end{equation*}
    Consider also the four points: 
       \begin{equation*}
        P_{ij} := (\alpha_{ij},\beta_{ij}):= L_i\cap L_j'
    \end{equation*}
    Then $\text{max}\{\beta_{11},\beta_{12}\} < \text{min}\{\beta_{21},\beta_{22}\}$
\end{lemma}
\begin{proof}
    Notice that 
    \begin{equation}
    \begin{split}
        &        P_{12} = \left(\alpha_{11}-\frac{\tau}{2}, \beta_{11} + \frac{\tau}{2}\right), \\
        &P_{21} = \left(\alpha_{11}+\frac{(Q'-Q)\lambda}{2}, \beta_{11} + \frac{(Q'-Q)\lambda}{2}\right), \\
        &        P_{22} = \left(\alpha_{11}-\frac{\tau}{2}+\frac{(Q'-Q)\lambda}{2}, \beta_{11} + \frac{\tau}{2}+\frac{(Q'-Q)\lambda}{2}\right)\\
    \end{split}
    \end{equation}
    The result follows since $\lambda > \tau$. 
\end{proof}
\begin{remark}
    By Lemma \ref{lemma_gap1}, the first failure point lies on one of the lines $L_1' $ or $L_2'$, corresponding to $\epsilon \in \{1,2\}$. Lemma \ref{worstcase} says that the first failure point is among the solutions to the equation 
    \[
2\gamma +\epsilon\tau = t\lambda +L.
\]
with the smallest possible value of $t$. So our choice for $L$ should minimize the value of $t$ for which there is a solution $(\gamma, \gamma + \epsilon \tau)$ to the above equation. 
    
\end{remark}

Using Lemmas \ref{lemmarop}, \ref{lemma_gap1} we can write the failure point as $\fp =(\gamma, \gamma + \epsilon \tau)$, where $\epsilon\in \{1,2\}$. Then condition 1 fails when 
\[
2\gamma +\epsilon\tau = t\lambda +L.
\]
We examine some possible choices for $L$ in the range $0\leq L \leq \lambda-1$. The best choice for $L$ is the largest possible, under the constraint of maximising the minimal solution for $t$ in the above equation. Note that with these choices for $L$, there is never a solution with $t<0$.
\begin{itemize}
    \item Suppose $\lambda > 2\tau$. Then there is always a solution on the line $t=1$. Consider the line $t=0$, equivalently solutions to the equation 
    \begin{equation*}
        2\gamma + \epsilon\tau = L.
    \end{equation*}
    \begin{itemize}
        \item Suppose $L$ is odd. It is important to select $L$ as large as possible, while avoiding a solution on the line $t=0$. By parity, it follows that $\epsilon =1$, and thus there is no solution on the line $t=0 $ so long as $L<\tau$. Thus, the maximum $L$ is $\tau-2$. The first failure point lies on the line $t=1$, with $\epsilon=2$, and it is $\left(\frac{\lambda-\tau-2}{2},\frac{\lambda+3\tau-2}{2}\right)$.

        \item Suppose $L$ is even. By parity, it follows that $\epsilon =2$, and thus there is no solution on the line $t=0 $ so long as $L<2\tau$. Thus, the maximum $L$ is $2\tau-2$. The first failure point lies on the line $t=1$, with $\epsilon=1$, and it is $\left(\frac{\lambda+\tau-2}{2},\frac{\lambda+3\tau-2}{2}\right)$.
    \end{itemize}
    \item Suppose $\lambda <2\tau.$
    \begin{itemize}
        \item Suppose $L$ is even. Since $\lambda>\tau$, there is always a solution on $t=1$ with $\epsilon=1$. To avoid a solution on $t=0$, we need to take $L<2\tau$,which is already satisfied by our choices of $L$. Therefore, the best choice for $L$ will be $\lambda-1$. The first failure point lies on the line $t=1$, with $\epsilon=1$, and it is  $\left(\frac{2\lambda-\tau-1}{2},\frac{2\lambda+\tau-1}{2}\right)$.

        \item Suppose $L$ is odd. If we pick $L < \tau$, then there is no solution on the line $t=0$. If we pick $L<2\tau-\lambda$ then there is no solution for $t=1$. There will always be a solution for $t=2$, with $\epsilon =1$. Thus, by selecting $L= 2\tau-\lambda-2$, the first failure point will be on the line $t=2$ with $\epsilon=1$. Solving, we find the first failure point to be $\left(\frac{\lambda+\tau-2}{2},\frac{\lambda+3\tau-2}{2}\right)$.
        Note that this choice for $L$ will be equivalent to $L=2\tau-2$. 
    \end{itemize}
\end{itemize}
In both subcases, it is optimal to select $L= 2\tau-2$. The resulting first failure point is 
\[
\boxed{\left(T_1,T_2\right) = \left(\frac{\lambda+\tau-2}{2},\frac{\lambda+3\tau-2}{ 2}\right)}
\]
and the greatest possible minimum distance of the code is 
\[
\boxed{T = \frac{\lambda+3\tau}{2}}.
\]
\subsection{Summary}\label{section_summary}
Here is a table summarising the results of the analysis of failure points.

\begin{center}
\begin{tabular}{||c| c| c |c |c||} 
 \hline
 Case & Conditions & Optimal $L$ & $(T_1,T_2)$ & $T$\\  
 \hline\hline
 1 & $\lambda$ even 
 & $2\tau-2$ 
 &$\left(\frac{\lambda-2}{2},\frac{\lambda+4\tau-2}{2}\right)$ 
 &  $\frac{\lambda+4\tau}{2}$\\ 
 \hline
 2 & \specialcell{$\lambda$ odd, and either $\lambda<\tau$, \\ $\tau$ even or $\rho=2$}
 & $\tau-2$  
 & $\left(\lambda-1,\lambda+\tau-1\right)$ 
 & $\lambda + \tau$ \\
 \hline
 3  & $ \lambda $ odd, $\lambda >\tau, \tau $ odd, $\rho \neq 2$ 
 & $2\tau-2$ 
 & $\left(\frac{\lambda+\tau-2}{2},\frac{\lambda+3\tau-2}{2}\right)$ 
 & $\frac{\lambda+3\tau}{2}$\\

 \hline
\end{tabular}
\end{center}
In all of the above cases, the length of our code will be $n = \lambda\tau\sigma$. We can then select 
$2\le d \le T$, and we will get a $[[n,n-2d+2,d]]_q$ quantum stabilizer code, by evaluating with the monomials $\{1,X,X^2,\ldots ,X^{d-2}\}$.

\section{Main Theorem}\label{mainsection}

Here is our main theorem, which is the construction outlined in Section 3, 
and uses the analysis of failure points from the previous section.

\begin{theorem}\label{main}
Let $q\ge 3$ be a prime power. Let $\lambda>1$ be a divisor of $q-1$, and 
let $\tau>1$ and  $\rho>1$ be divisors of $q+1$. 
Assume that $\gcd (\lambda,\tau)=1$. 
Let  $\kappa=\gcd (\lambda,\rho)\cdot  \gcd(\tau,\rho)$ and assume that $ \frac{\rho}{\kappa} \ge 2$.
Let $\sigma$ be any integer with $ \frac{\rho}{\kappa} \ge  \sigma \ge 2$.
Let $n = \lambda\tau\sigma$.
Let $T$ be chosen according to this table:
\begin{center}
\begin{tabular}{||c|c||} 
 \hline
  Conditions &  $T$ \\  
 \hline\hline
  $\lambda$ even &  $\frac{\lambda+4\tau}{2}$\\ 
 \hline
  \specialcell{$\lambda$ odd, and either $\lambda<\tau$, \\ $\tau$ even or $\rho=2$} &  $\lambda + \tau$   \\
 \hline
  $ \lambda $ odd, $\lambda >\tau, \tau $ odd, $\rho \neq 2$ &  $\frac{\lambda+3\tau}{2}$ \\

 \hline
\end{tabular}
\end{center}

Then for any $d$ with  $2\le d \le T$ there exists a $[[n,n-2d+2,d]]_q$ quantum MDS  code.
\end{theorem}

\begin{proof}
By selecting $L$ in correspondence with $T$ from the table in Section \ref{section_summary}, setting $k:=d-1$, and selecting $\boldsymbol{A}, \boldsymbol{v}$ as in Section \ref{section_code_construction}, we  obtain a GRSC $C$ from Definition \ref{def_GRSC}, namely 
\[
C:= \text{GRS}_{n,k}(\boldsymbol{v},\boldsymbol{A}).
\]
By Remark \ref{remark_GRSC}, the length and dimension of this code will be $n,k$ respectively. By the same remark, the minimum distance of $C^{\perp_h}$ will be $d$. By Theorem \ref{thm_OC} and the analysis in Section \ref{section_first_failure_point}, we  deduce that $C \subseteq C^{\perp_h}$.  

Thus, by Theorem \ref{thm_quantcode},  there exists a $[[n,n-2d+2,d]]_q$ quantum MDS  code.
\end{proof}

\section{Examples}\label{examples}

We present some applications of our theorem, which consist of certain choices of 
the parameters.
Usually we write the dimension as $k$, but since we give $n$ and $d$
then $k$ is determined by the formula $n+2=k+2d$.

\subsection{New Families}
We will demonstrate that the construction in this paper yields new codes by comparing some sample 
parameters with the table of known parameters in the very recent paper \cite{wan2024}.
\subsubsection{Example }

\begin{corollary}\label{c1}
Let $q \equiv 3 \pmod{8}$, $q>3$. Then for any $2 \leq d \leq \frac{5q+1}{8}$, there exists a 
$[[\frac{3(q^2-1)}{8},k,d]]_q$ quantum MDS  code.
\end{corollary}

\begin{proof}
    Choose $\lambda=\frac{q-1}{2}$, $\tau=\frac{q+1}{4}$, $\rho=4$. Then $\kappa=1$, so we choose $\sigma =3$.
Then we are in case 3, and so $T=(\lambda+3\tau)/2$.
\end{proof}

We remark that the length is not a multiple of $q-1$, and not a multiple of $q+1$, and $d>q/2$.

This family of codes is new. For example, when $q=11$ we get a $[[45, 33,7]]_{11}$ MDS code from Corollary \ref{c1}. 
We will now explain why this code  is not covered by any of the constructions listed in the tables in \cite{wan2024}.

In Table 1 of \cite{wan2024} the length $n$ is divisible by $q+1$, which is not the case for our code.
In rows 1-15 of Table 2 of \cite{wan2024}, 
the length $n$ satisfies $n\le q+1$, or $n$ divides $q^2+1$,
or $n$ satisfies one of the
following congruences: 
$n \equiv 0,1,2 \pmod{q+1}$, or $n\equiv 0 \pmod{q}$, or
$n \equiv 0,1 \pmod{q-1}$.
The same holds for the new codes in Table 3 of \cite{wan2024}.
Our $[[45, 33,7]]_{11}$ MDS code does not fit any of these categories.
It is easy to check that our code does not fit rows 16,17,18. 
The codes in rows 19 and 21 have even length, so our code is not in those categories.
Finally, checking all the possibilities for $s$ and $t$ in row 20 does not result in a code
of length 45 when $q=11$.

This rules out all previous constructions, and hence Corollary \ref{c1} gives new codes.
Therefore Theorem \ref{main} gives new codes.

\subsubsection{Example}

\begin{corollary}\label{c2}
Let $q$ be odd. Let $m$ be a divisor of  $\frac{q+1}{2}$
such that $1\le m<\frac{q+1}{2}$ and $\frac{q+1}{2m}$ is even.
Let $2\leq \sigma \leq 2m$. Then for any $2\leq d \leq \frac{q-1}{2}+\frac{q+1}{2m}$, there exists a 
$[[\sigma\frac{(q^2-1)}{4m},k,d]]_q$ quantum MDS  code.
\end{corollary}

\begin{proof}
The hypotheses imply $q \equiv 3 \pmod{4}$.
    Choose $\lambda=\frac{q-1}{2}$, $\tau=\frac{q+1}{2m}$, $\rho=q+1$. 
    Then $\lambda$ is odd and $\tau$ is even, so we are in case 2.
    Then $\kappa = \frac{q+1}{2m}$ so $\rho/\kappa=2m$, so we are free to choose $2 \leq \sigma \leq 2m$.  We are in case 2,  so $T=\lambda + \tau$.
\end{proof}

We remark that the code length is always a multiple of $q-1$, is not a multiple
of $q+1$, and that $d>q/2$.

When $\sigma=2$, this family is new. For example, consider $q=83, m=7$. This results in a $[[492,400,47]]_{83}$ MDS code from Corollary \ref{c2}. We will now explain why this code  is not covered by any of the constructions listed in the tables in \cite{wan2024}.

Our code length is neither $0,1$ or $2$ (mod $q+1$), nor $0$ (mod $q$), nor $\leq q+1$, nor is it a multiple of $q^2+1$. By inspection, the code does not fit rows 16,17 or 18 of Table 2. The length does satisfy $n \equiv 0 $ (mod $q-1$), leaving us to compare with rows 12-15 and 19-21 of Table 2, and row 2 of Table 3. 

Writing the $[[492,400,47]]_{83}$ code in the form $n= r\frac{q^2-1}{h}$ with $h\mid (q+1)$, the only valid possibilities for $(r,h)$ to obtain our length are $(1,14), (2,28), (3,42), (6,84)$. Checking rows 12-15 of Table 2 and row 2 of Table 3, the largest obtainable distance is 45. 

In row 19, the only possibilities for $t$ are $2,82$. If $t=82$, then the distance is $\leq (\frac{t}{2}+1)(\frac{q-1}{t})+1=42$. If $t=2$, then the length will be $>2000$, following from the condition $rl\frac{q^2-1}{st}<q-1$. In row 20, the length must be odd. In row 21, the length cannot be a multiple of $q-1$ unless $l=\frac{t}{2}$, in which case $n>3444$. 

This rules out all previous constructions, and hence Corollary \ref{c2} gives new codes.

\subsubsection{Example}

\begin{corollary}\label{c3}
Let $q \equiv 5 \pmod{8}$. Let $m\mid \frac{q+1}{2}, 1<m<\frac{q+1}{2}$. Let $2\leq \sigma \leq m$. Then for any $2 \leq d \leq \frac{q-1}{2}+\frac{q+1}{m}$, there exists a 
$[[\sigma\frac{(q^2-1)}{2m},k,d]]_q$ quantum MDS  code.
\end{corollary}

\begin{proof}
    Choose $\lambda=q-1$, $\tau=\frac{q+1}{2m}$, $\rho=q+1$. Then $\kappa = \frac{q+1}{m}$, so we are free to choose $2\leq \sigma \leq m$. We are in case 1, and so $T=\frac{\lambda+4\tau}{2}$.
\end{proof}

We remark that the code length is always a multiple of $q-1$, is not a multiple
of $q+1$, and that $d>q/2$.

When $\sigma=2$, this family is new. For example, consider $q=29, m=3$. This results in a $[[280,234,24]]_{29}$ MDS code from Corollary \ref{c3}. We will now explain why this code  is not covered by any of the constructions listed in the tables in \cite{wan2024}. For the same reasons as outlined after Corollary \ref{c2}, we only need to compare with rows 12-15 and 19-21 of Table 2, and row 2 of Table 3. 

Writing the $[[280,234,24]]_{29}$ code in the form $n= r\frac{q^2-1}{h}$ with $h\mid (q+1)$, the only valid possibilities for $(r,h)$ to obtain our length are $(1,3), (2,6), (5,15), (10,30)$. Checking rows 12-15 of Table 2 and row 2 of Table 3, only row 14 applies and the largest obtainable distance is 19. 

In rows 19, $t$ must be an even divisor of $q-1$. If $t\ge 4$ then the distance is $\leq 22$. If $t=2$, then the length will be $>300$, following from the condition $rl\frac{q^2-1}{st}<q-1$. The same reasoning applies to row 21. Finally, in row 20, the length must be odd. 

This rules out all previous constructions, and hence Corollary \ref{c3} gives new codes.

\subsection{Matching Families}

In this subsection we will give a few examples that show that our Theorem \ref{main} constructs codes that 
are already in other recent papers.

\subsubsection{Example}

\begin{corollary}
Let $q \equiv 3 \pmod{4}$. Let m be an odd divisor of $q+1$. 
Then for any $2 \leq d \leq \frac{q+4m-1}{2}$, there exists a $[[2m(q-1),k,d]]_q$ quantum MDS  code.
\end{corollary}

\begin{proof}
    If $m \neq 1$, choose $\lambda=q-1$, $\tau=m$, $\rho=4$, $\sigma =2$.
Then we are in case 1, and so $T=(\lambda+4\tau)/2=\frac{q+4m-1}{2}$.

    If $m = 1$, choose $\lambda=\frac{q-1}{2}$, $\tau=2$, $\rho=4$, $\sigma =2$. 
Then we are in case 2, and so $T=\lambda+\tau=\frac{q+3}{2}$.

\end{proof}
The length here is a multiple of $q-1$.
These parameters generalise those in Theorem 4.5 of \cite{ApplicationConst2015}, achieving the same length but a larger minimum distance.

\subsubsection{Example}

\begin{corollary}
Let $q \equiv 3 \pmod{8}$. 
Let $m$ be an odd divisor of $q-1$.
Then for any $2 \leq d \leq \frac{q+4m-1}{2} $, there exists a $[[m(q+1),k,d]]_q$ quantum MDS  code.
\end{corollary}

\begin{proof}
    Choose $\lambda=2m$, $\tau=\frac{q+1}{4}$, $\rho=4$, $\sigma =2$.
Then we are in case 1, and so $T=(\lambda+4\tau)/2=m+\frac{q+1}{2}$.
\end{proof}
The length here is a multiple of $q+1$.
These parameters are the same as in Theorem 3.7 of \cite{Constacyclic2014}, when $q \equiv 3 $ (mod $8$)

\subsubsection{Example }

\begin{corollary}
Let $q \equiv 11 \pmod{24}$. Then for any $2 \leq d \leq \frac{q+11}{2}$ there exists a 
$[[6(q-1),k,d]]_q$ quantum MDS  code.
\end{corollary}

\begin{proof}
Choose $\lambda=q-1$, $\tau=3$, $\rho=4$, $\sigma =2$.
Then we are in case 1, and so $T=(q-1)/2+6$. 
\end{proof}

These parameters are already known, we match Theorem 3.3 of \cite{Wang2014}.
The length is a multiple of $q-1$ here.

%When $q=11$ we get a $[[60, 40,11]]_{11}$ MDS code.

\subsection[Examples with small d]{Examples with small $d$}

In this section we give some applications of Theorem \ref{main} in a different direction. 
We fix the distance $d$ to be small, say 5 or 7.

\begin{corollary}
Let $q \equiv 1 \pmod{6}$. 
Then there exists a 
$[[6\sigma,k,5]]_q$ quantum MDS  code, for any $\sigma$
with $2 \le \sigma \le (q+1)/2$.
\end{corollary}

\begin{proof}
    Choose $\lambda=3$, $\tau=2$, $\rho=q+1$. 
    Then $\kappa = \gcd(\lambda,\rho) \cdot \gcd(\tau,\rho)=2$ and 
    we can let $\sigma$ be
    any integer between 2 and $(q+1)/2$.
We are in case 2, and so $T=\lambda+\tau=5$, and the result follows from Theorem \ref{main}.
\end{proof}

When $q=7$ we get  $[[12,4,5]]_{7}$, $[[18,10,5]]_{7}$, $[[24,16,5]]_{7}$, codes.
When $q=13$ we get codes with those parameters, plus 
$[[30,22,5]]_{13}$, $[[36,28,5]]_{13}$, $[[42,34,5]]_{13}$  codes.
For $q>13$ and $q \equiv 1 \pmod{6}$ we obtain codes with these same parameters,
and more codes.

\begin{corollary}
Let $q\equiv 5 \pmod6$, $q>5$.
Then there exists a 
$[[6\sigma,k,7]]_q$ quantum MDS  code, for any $\sigma$
with $2 \le \sigma \le (q+1)/6$.
\end{corollary}

\begin{proof}
    Choose $\lambda=2$, $\tau=3$, $\rho=q+1$.
    Then $\kappa = \gcd(\lambda,\rho) \cdot \gcd(\tau,\rho)=6$ and  $\rho/\kappa=(q+1)/6$.
Then we are in case 1, and so $T=(\lambda+4\tau)/2=7$.
\end{proof}

When $q=11$ we get a $[[12,0,7]]_{11}$ code.
When $q=17$ we get a $[[12,0,7]]_{17}$ code, and a $[[18,6,7]]_{17}$ code.
When $q=23$ we get  $[[12,0,7]]_{23}$, $[[18,6,7]]_{23}$, $[[24,12,7]]_{23}$ codes.

We give one more example of this type.

\begin{corollary}
Let $q \equiv 7 \pmod{12}$. 
Then there exists a 
$[[12\sigma,k,7]]_q$ quantum MDS  code, for any $\sigma$
with $2 \le \sigma \le (q+1)/2$.
\end{corollary}

\begin{proof}
    Choose $\lambda=3$, $\tau=4$, $\rho=q+1$.
    Then $\kappa = \gcd(\lambda,\rho) \cdot \gcd(\tau,\rho)=\tau=4$ so $\rho/\kappa=(q+1)/4$.
Then we are in case 2, and so $T=\lambda+\tau=7$.
\end{proof}

For $q=7$ we get a $[[24,12,7]]_{7}$ code.
For $q=19$ we get a $[[24,12,7]]_{19}$ code and a $[[36,24,7]]_{19}$ code.

Many more examples like these can be constructed by similar choices of small values of $\lambda$ and $\tau$,
ensuring that the hypotheses of Theorem \ref{main} hold.

%Notice that we can puncture twice, and we get  $[[12\sigma -2,k,5]]_q$ MDS code.
%For example, $[[24,12,7]]_{13}$ gives $[[22,14,5]]_{13}$ after puncturing twice.
%Also $[[23,13,6]]_{13}$ by puncturing once.

\section*{Acknowledgements}

This publication has emanated from research conducted with the financial support of Science Foundation Ireland under Grant number 21/RP-2TF/10019 for the first author. 
For the purpose of Open Access, the author has applied a CC BY public copyright licence to any Author Accepted Manuscript version arising from this submission.

The second  author is partially supported by  MICIN/AEI  Grants PID2022-138906NB-C22, as well as by Universitat Jaume I, Grants UJI-B2021-02, GACUJIMA/2023/06.

\section*{Data Availability} 
Data sharing is not applicable to this article as no datasets were generated or analyzed
during the current study.

\section*{Conflict of interest} 
We declare no conflicts of interest.

\bibliographystyle{plain}
\bibliography{biblio}
\end{document}